\documentclass[]{IEEEtran}

\usepackage[compress]{cite}
\usepackage{amsmath}
\usepackage{graphicx}
\usepackage{physics}
\usepackage{amssymb}

\newenvironment{proof}{\paragraph{Proof:}}{\hfill$\square$}

\newtheorem{theorem}{\bf Theorem}
\newtheorem{proposition}{\bf Proposition}

\newtheorem{corollary}{\bf Corollary}

\newtheorem{remark}{\bf Remark}
\newtheorem{lemma}{\bf Lemma}
\newtheorem{definition}{\bf Definition}

\allowdisplaybreaks

\renewcommand{\H}{\mathcal{H}}
\renewcommand{\L}{\mathcal{L}}
\renewcommand{\S}{\mathcal{S}}
\renewcommand{\P}{\mathcal{P}}
\newcommand{\R}{\mathbb{R}}

\begin{document}

\title{Sample Complexity Bounds for Scalar Parameter Estimation Under Quantum Differential Privacy } 

\author{Farhad Farokhi
\thanks{F.~Farokhi is with the Department of Electrical and Electronic Engineering at the University of Melbourne.}
}

\maketitle

\begin{abstract}
This paper presents tight upper and lower bounds for minimum number of samples (copies of a quantum state) required to attain a prescribed accuracy (measured by error variance) for scalar parameters estimation using unbiased estimators under quantum local differential privacy for qubits. Particularly, the best-case (optimal) scenario is considered by minimizing the sample complexity over all differentially-private channels; the worst-case channels can be arbitrarily uninformative and render the  problem ill-defined. In the small privacy budget $\epsilon$ regime, i.e., $\epsilon\ll 1$, the sample complexity scales as $\Theta(\epsilon^{-2})$. This bound matches that of classical parameter estimation under local differential privacy. The lower bound however loosens in the large privacy budget regime, i.e., $\epsilon\gg 1$. The upper bound for the minimum number of samples is  generalized to qudits (with dimension $d$) resulting in sample complexity of $\mathcal{O}(d\epsilon^{-2})$. 
\end{abstract}


\section{Introduction}
Differential privacy~\cite{dwork2006calibrating,dwork2014algorithmic} has taken over computer science literature as the gold standard for private data analysis with deployment in government and commercial services~\cite{abowd2018us,thakurta2017learning, erlingsson2014rappor}. Differential privacy requires that changes in a single entry of a dataset to only generate small changes in the probability distribution of the output~\cite{dwork2014algorithmic}. The magnitude of the change in the probability distribution of the output is parameterized by the so-called privacy budget, which is a design choice and relates to the risk appetite of the data curator. Differential privacy is often guaranteed by adding noise via Laplace, Gaussian, or exponential mechanisms~\cite{dwork2014algorithmic} to the outputs, where the magnitude of the noise is proportional to the sensitivity of the output to individual data changes and reciprocal to the privacy budget. In the large-data regime, the effect of the noise can be mostly removed from empirical distributions via deconvolution~\cite{farokhi2020deconvoluting}, but differential privacy is known to cause utility degradation in general~\cite{seeman2024between}. The widespread adoption of differential privacy has resulted in natural questions involving sample complexity of learning and estimation problems, such as hypothesis testing~\cite{pensia2023simple}, parameter estimation~\cite{cai2021cost}, mean estimation~\cite{liu2021robust}, classification~\cite{feldman2014sample}, and regression~\cite{wu2020value}, under privacy constraint. 

In light of advances in quantum computing and communication, differential privacy has been extended to the quantum domain~\cite{aaronson2019gentle,zhou2017differential, hirche2023quantum}. Similarly, quantum differential privacy requires that changes in a single entry of the dataset to only result in small changes in the probability distribution of measurements implemented on encoded quantum states~\cite{hirche2023quantum}, which can be guaranteed via quantum noise, such as depolarizing noise~\cite{zhou2017differential,hirche2023quantum}. Further extensions, such as pufferfish privacy~\cite{nuradha2024quantum} and information-theoretic privacy~\cite{farokhi2024barycentric}, have been also presented. These definitions have fueled a line of research in understanding fundamental limits of quantum data processing under privacy constraints. Hypothesis testing under quantum differential privacy was recently studied in~\cite{farokhi2023quantum, nuradha2025contraction, cheng2024sample}. Limits of quantum machine learning under privacy were also studied in~\cite{watkins2023quantum, farokhi2024barycentric}.

This paper focuses on deterministic (non-Bayesian) parameter estimation under quantum differential privacy. We particularly consider the case where quantum states depend on scalar parameters that must be estimated. For instance, the parameter can model an unknown aspect of Hamiltonian in quantum computing and the state is the outcome of the evolution under the-said Hamiltonian with fixed initial condition. In this case, estimating the parameter is akin to reverse-engineering the quantum computing circuit. Alternatively, in quantum machine learning, the state can encode some sensitive data and the parameter can be a private attribute that the adversary may seek to learn. Here, we are interested in determining the minimum number of samples, or copies of quantum states, required to attain a prescribed parameter estimation accuracy, measured by the variance of the estimation error, in the best-case (optimal) scenario. That is, we characterize the smallest number of samples over all differentially-private channels. This is because the worst-case differentially-private channels can be arbitrarily uninformative and render the  problem ill-defined; see Footnote~\ref{foot:1}. The minimum number of samples is referred to as the sample complexity. From the perspective of the estimator or adversary, the sample complexity determines the amount of the data that needs to be gathered for reliable parameter estimation while, from the perspective of the data curator, the sample complexity provides bounds on the number of interaction rounds with an adversary without enabling parameter estimation beyond a prescribed level. The latter provides an accuracy or risk interpretation for the privacy budget in quantum differential privacy. We use the quantum hockey-stick divergence~\cite{hirche2023quantum} to capture quantum differential privacy and the quantum Cram\'{e}r-Rao bound~\cite{helstrom1967minimum, braunstein1994statistical, braunstein1996generalized} to establish bounds on the number of quantum state copies or samples required to attain a prescribed estimation error variance. We particularly use the Bloch sphere representation and explicit Fisher information formulas in this regime~\cite{zhong2013fisher}. The quantum Cram\'{e}r-Rao bound, in its vanilla form, restricts us to unbiased estimators. The motivation behind focusing on scalar parameters stems from the tightness of quantum Cram\'{e}r-Rao bound in the finite measurement regime for scalar parameters. 

The remainder of this paper is organized as follows. We first present some definitions and preliminary results in Section~\ref{sec:prelim}. The  sample complexity results are presented in Section~\ref{sec:results}. Finally, Section~\ref{sec:conc} concludes the paper and presents some directions for future research.

\section{Preliminary Material} \label{sec:prelim}

\subsection{Complexity Notation}
In this paper, we use complexity notation to concisely present the results in the case of small privacy budget $\epsilon \ll 1$. The privacy budget is made precise in the following subsections. We write $f(\epsilon)=\mathcal{O}(g(\epsilon))$ if there exist $\epsilon_0$ and $c>0$ such that $|f(\epsilon)|\leq c |g(\epsilon)|, \forall \epsilon\in(0,\epsilon_0)$. This is referred to as the big O notation and essentially establishes an asymptotic upper bound. For instance, $f(\epsilon)=\mathcal{O}(\epsilon^{-1})$ means $|f(\epsilon)|$ is upper bounded by $c/\epsilon$ for some appropriately chosen constant $c>0$ over a small enough region around zero. We write $f(\epsilon)=\Omega(g(\epsilon))$ if there exist $\epsilon_0$ and $c>0$ such that $|f(\epsilon)|\geq c |g(\epsilon)|, \forall \epsilon\in(0,\epsilon_0)$. This is referred to as the big Omega notation and establishes an asymptotic lower bound. For instance, $f(\epsilon)=\Omega(\log(\epsilon))$ implies that $|f(\epsilon)|$ is lower bounded by $c\log(\epsilon)$ for some appropriately chosen constant $c>0$ over a small enough region around zero. Finally, we write $f(\epsilon)=\Theta(g(\epsilon))$ if $f(\epsilon)=\mathcal{O}(g(\epsilon))$ and $f(\epsilon)=\Omega(g(\epsilon))$. This implies that $g$ sandwiches $f$ tightly (in terms of order, not constants).

\subsection{Density Operators}
The following definitions and preliminary results are adopted from~\cite{wilde2013quantum}. The set of linear operators from finite-dimensional Hilbert space $\H$ to itself is denoted by $\L(\H)$. The set of positive semi-definite linear operators is denoted by $\P(\H)\subset\L(\H)$. The set of density operators (i.e., positive semi-definite linear operators with unit trace) is denoted by $\S(\H)\subset\P(\H)$. 
Qubits, which stand for quantum bits, are  basic units of quantum information correspond to 2-dimensional Hilbert spaces. In the so-called Bloch sphere representation~\cite[p.\,105]{nielsen2010quantum}, the density operator $\rho$ for any qubit can be represented as
\begin{align}\label{eqn:qubit}
	\rho=\frac{1}{2}\left(I+\omega.\hat{\sigma}\right),
\end{align}
where $\omega\!=\!(\omega_x,\omega_y,\omega_z)\!\in\!\R^3$ is such that $\|\omega\|_2\!\leq\! 1$ (with $\|\omega\|_2^2\!=\!\omega^\top \omega$) and $\hat{\sigma}\!=\!(\hat{\sigma}_x,\hat{\sigma}_y, \hat{\sigma}_z)$ is the tuple of Pauli matrices
\begin{align*}
	\hat{\sigma}_x
	:=
	\begin{bmatrix}
		0 & 1\\
		1 & 0
	\end{bmatrix},\quad 
	\hat{\sigma}_y
	:=
	\begin{bmatrix}
		0 & -i\\
		i & 0
	\end{bmatrix},\quad 
	\hat{\sigma}_z
	:=
	\begin{bmatrix}
		1 & 0\\
		0 & -1
	\end{bmatrix}.
\end{align*}
Here, these matrices are represented in the so-called computational basis. Note that, in the Bloch sphere representation, the definition of the inner product is expanded to allow for $\omega.\hat{\sigma}
:=\omega_x\hat{\sigma}_x
+\omega_y\hat{\sigma}_y+
\omega_z\hat{\sigma}_z.$
A quantum channel, in its most general form, is a mapping on the space of density operators that is both completely positive and trace preserving. In the case of qubits, for each quantum channel $\mathcal{E}:\S(\H)\rightarrow\S(\H)$, there exist $A\in\R^{3\times 3}$ and $c\in\R^3$ such that 
\begin{align}
	\mathcal{E}(\rho)=\frac{1}{2}\left(I+(A\omega+c).\hat{\sigma}\right).\label{eqn:equivalence}
\end{align}
Note that it must be that $\|A\omega+c\|_2\leq 1$ for all $\omega$ such that $\|\omega\|_2\leq 1$. This is to ensure that the output $\mathcal{E}(\rho)$ is still a density operator. A necessary condition for this is that $\|c\|_2\leq 1$ (because $1\geq \|A\omega+c\|_2$ for $\omega=0$) and $\|A\|_2=\sigma_{\max}(A)\leq 2$ (because $1\geq \|A\omega+c\|_2\geq \|A\omega\|_2-\|c\|_2$).  Qudits are extensions of qubits to $d$-dimensional Hilbert spaces. In the (generalized) Bloch sphere representation (which is no longer of interest for visualization)~\cite{zhong2013fisher}, qudits can be represented as
\begin{align*}
	\rho=\frac{1}{d}I+\frac{1}{2}\omega.\hat\eta,
\end{align*}
where $\omega=(\omega_i)_1^{d^2-1}\in\R^{d^2-1}$ with $\|\omega\|_2\leq \sqrt{2(d-1)/d}$ and $\hat\eta=(\hat\eta_i)_1^{d^2-1}$ is the generator (i.e., orthonormal Hermitian operators with trace zero) of the Lie algebra $\mathfrak{su}(d)$. In the case of qudits, for each quantum channel $\mathcal{E}:\S(\H)\rightarrow\S(\H)$, there exist $A\in\R^{d^2-1\times d^2-1}$ and $c\in\R^{d^2-1}$ such that 
\begin{align}
	\mathcal{E}(\rho)=\frac{1}{d}I+\frac{1}{2}(A\omega+c).\hat\eta,
	\label{eqn:equivalence_2}
\end{align}
where, again by definition, $\|A\omega+c\|_2\leq \sqrt{2(d-1)/d}$ for all $\omega$ such that $\|\omega\|_2\leq \sqrt{2(d-1)/d}$. Given the equivalence relationships in~\eqref{eqn:equivalence} and~\eqref{eqn:equivalence_2}, we may abuse the notation by referring to quantum channel $\mathcal{E}$ with $(A,c)$.

\subsection{Quantum Fisher Information}
The following definitions and preliminary results are adopted from~\cite{zhong2013fisher}. Let density operator $\rho_\lambda\in\S(\H)$ depend on scalar parameter $\lambda\in\R$. Assume that $\rho_\lambda$ is continuously differentiable with respect to $\lambda$. The quantum Fisher information is
\begin{align}
	\mathcal{F}(\rho_\lambda)
	:=\trace(\rho_\lambda L_\lambda^2)
	=\trace\left(\left( \frac{\partial}{\partial\lambda} \rho_\lambda \right)L_\lambda\right),
\end{align}
where symmetric logarithmic derivative operator $L_\lambda\in\L(\H)$ is any Hermitian operator that satisfies
\begin{align*}
	\frac{\partial}{\partial\lambda} \rho_\lambda
	=\frac{1}{2}\left(\rho_\lambda L_\lambda+ L_\lambda \rho_\lambda\right).
\end{align*}
For qubits, this definition can be simplified\footnote{A keen reader may note that we have already used $\omega_x$, $\omega_y$, and $\omega_z$ to denote the elements of vector $\omega\in\R^3$ for qubits. For qudits, we used $\omega_i$ to denote the elements of $\omega\in\R^{d^2-1}$. By writing $\omega_\lambda$ to emphasize the dependence of $\omega$ on the parameter $\lambda$, we may run the risk of notation overload and confusion. However, to avoid confusion, in this paper, we use lowercase Greek letters, such as $\lambda$, to denote parameters and lowercase Roman letters and numbers to denote the coordinates.} to 
\begin{align}
	\mathcal{F}(\rho_\lambda)
	=
	\begin{cases}
		\|\partial_\lambda \omega_\lambda\|_2^2+\displaystyle \frac{|\braket{\omega_\lambda}{\partial_\lambda \omega_\lambda}|^2}{1-\|\omega_\lambda\|_2^2}, & \|\omega_\lambda\|_2<1,\\
		\|\partial_\lambda \omega_\lambda\|_2^2, & 
		\|\omega_\lambda\|_2=1,
	\end{cases}
\end{align}
where $\partial_\lambda \omega_\lambda=\partial \omega_\lambda/\partial \lambda$. Note that the quantum Fisher information is not necessarily continuous everywhere (particularly as $\|\omega_\lambda\|_2\rightarrow 1$)~\cite{vsafranek2017discontinuities}. For qudits, this definition can be simplified to 
\begin{align}
	\mathcal{F}(\rho_\lambda)
	\!=\!\!
	\begin{cases}
		(\partial_\lambda \omega_\lambda)^\top \mathcal{M}(\omega_\lambda)^{-1} (\partial_\lambda \omega_\lambda)^\top, & \!\!\!\!\|\omega_\lambda\|_2\!<\!\sqrt{\frac{2(d-1)}{d}},\\
		\|\partial_\lambda \omega_\lambda\|_2^2, & 
		\!\!\!\!\|\omega_\lambda\|_2\!=\!\sqrt{\frac{2(d-1)}{d}},
	\end{cases}
\end{align}
where $\mathcal{M}(\omega_\lambda)=(2/d)I-\omega_\lambda\omega_\lambda^\top +G(\omega_\lambda)$ with element in the $i$-th row and $j$-th column of $G(\omega_\lambda)\in\R^{(d^2-1)\times (d^2-1)}$ being equal to $G_{ij}=\frac{1}{4}\sum_{k=1}^{d^2-1}  \trace\left((\hat{\eta}_i\hat{\eta}_j+\hat{\eta}_j\hat{\eta}_i)\hat{\eta}_k\right) (\omega_\lambda)_k.$
Assume that we can gather measurements from $N\geq 1$ copies of $\rho_\lambda$, denoted by $\rho_\lambda^{\otimes N}$, by implementing a positive operator-valued measure (POVM). The measurement outcomes can be used to estimate parameter $\lambda$. Let $\hat{\lambda}$ denote any unbiased estimate of the parameter $\lambda$, that is, $\mathbb{E}\{\hat{\lambda}\}=\lambda$, where expectation is taken with respect to the randomness of the quantum measurement. The so-called
quantum Cram\'{e}r-Rao theorem implies that
\begin{align}\label{eqn:CRB}
	\mathbb{E}\{(\lambda-\hat{\lambda})^2\}\geq \frac{1}{N \mathcal{F}(\rho_\lambda)}.
\end{align}
In the scalar parameter case discussed above, the lower bound can be saturated~\cite{braunstein1994statistical, braunstein1996generalized}; see~\cite{10480691,ragy2016compatibility} for generalized saturability results. 

\begin{remark}
	Extension to biased estimators would follow a similar line of reasoning as in unbiased estimators presented here. To illustrate this, consider the set of biased estimators such that $\mathbb{E}\{\hat{\lambda}\}=\lambda+\zeta(\lambda)$ with $|\partial \zeta(\lambda)/\partial \lambda|\leq b<1$. For such biased estimators, the quantum Cram\'{e}r-Rao theorem can be extended~\cite{kardashin2025predicting} to get
	\begin{align*}
		\mathbb{E}\{(\lambda-\hat{\lambda})^2\}
		&\geq \frac{|\partial \mathbb{E}\{\hat{\lambda}\}/\partial \lambda|^2}{N \mathcal{F}(\rho_\lambda)}+|\mathbb{E}\{\hat{\lambda}\}-\lambda|^2\geq \frac{(1-b)^2}{N \mathcal{F}(\rho_\lambda)}.
	\end{align*}
	This is similar in essence to~\eqref{eqn:CRB} and therefore can be used in the proofs (see Theorem~\ref{thm:parameter_complexity_qubit}) to derive complexity bounds for biased estimators. 
\end{remark}

\subsection{Quantum Differential Privacy}
The following definitions and preliminary results are adopted from~\cite{hirche2023quantum,angrisani2022quantum}. The quantum local differential privacy~\cite{angrisani2022quantum} is akin to quantum differential privacy with the exception of removing the so-called ``neighboring quantum states''. Local differential privacy is a stronger or more robust approach to privacy, removing the need for a trusted curator~\cite{duchi2013local,angrisani2022quantum}. 

\begin{definition}
	For $\epsilon\geq 0$, quantum channel $\mathcal{E}:\S(\H)\rightarrow \S(\H)$ is $\epsilon$-locally differentially private  if 
	\begin{align}
		\trace(M\mathcal{E}(\rho))\leq e^\epsilon \trace(M\mathcal{E}(\sigma)),
	\end{align}
	for all operators $0\preceq M\preceq I$, where $A\preceq B$ means $B-A\in\P(\H)$, and all density operators $\rho,\sigma\in\S(\H)$. The set of all quantum channel that are $\epsilon$-locally differentially private is denoted by ${\rm LDP}_{\epsilon}$.
\end{definition}

For density operators $\rho,\sigma\in\S(\H)$, the quantum hockey-stick divergence is 
\begin{align}
	E_\gamma(\rho\|\sigma)=\frac{1}{2}\|\rho-\gamma\sigma\|_1+\frac{1}{2}(1-\gamma),
\end{align}
where $ \|M \|_{1}:=\trace(|M|)$ is the trace norm of operator $M\in\L(\H)$ and $|M|=\sqrt{M^\dag M}$. 

\begin{lemma} \label{lem:iff_QLDP}
	Quantum channel $\mathcal{E}\in {\rm LDP}_{\epsilon}$ if and only if $E_{e^\epsilon}(\mathcal{E}(\rho)\|\mathcal{E}(\sigma))= 0$ for all $\rho,\sigma\in\S(\H)$.
\end{lemma}

\begin{proof}
	The proof follows from~\cite[Lemma~III.2]{hirche2023quantum} by setting $\delta=0$.
\end{proof}

We can prove the following lemma for differentially private quantum channels acting on qubits. This results, particularly the ``only if'' part, plays a pivotal role in establishing the sample complexity bounds in the next section. 

\begin{lemma} \label{lem:iff_QLDP_qubit}
	Let $\dim(\H)=2$. $(A,c)\!\in\!{\rm LDP}_{\epsilon}$ if~and~only~if
	\begin{align}
		\|A(\omega-\nu)+(1-e^\epsilon)(A\nu+c)\|_2\leq e^\epsilon-1,  
	\end{align}
	for all $\omega,\nu\in\R^3$ such that $\|\omega\|\leq 1$ and $\|\nu\|\leq 1$.
\end{lemma}

\begin{proof}
	Let $\rho=(I + \omega.\hat{\sigma})/2$ and $
	\sigma=(I+\nu.\hat{\sigma})/2$.
	Therefore, $\mathcal{E}(\rho)=(I + \bar{\omega}.\hat{\sigma})/2$ and $
	\mathcal{E}(\sigma)=(I+\bar{\nu}.\hat{\sigma})/2$,
	where $\bar{\omega}=(A\omega+c)$ and $\bar{\nu}=(A\nu+c)$. 
	Note that
	\begin{align}
		\|\mathcal{E}(\rho)-e^\epsilon\mathcal{E}(\sigma)\|_1
		=&
		\frac{1}{2}\left\|(1-e^\epsilon)
		+(\bar\omega-e^\epsilon\bar\nu).
		\hat{\sigma}\right\|_1
		\nonumber
		\\
		=&
		\frac{1}{2}\left|(1-e^\epsilon)+\|\bar\omega-e^\epsilon\bar\nu\|_2\right|
		\nonumber
		\\
		&+\frac{1}{2}\left|(1-e^\epsilon)-\|\bar\omega-e^\epsilon\bar\nu\|_2\right|,
		\label{eqn:proof:0}
	\end{align}
	where the second equality follows from Lemma~\ref{lem:useful} in the appendix. Because $e^\epsilon\geq 1$ (or equivalently $1-e^\epsilon\leq 0$) for all $\epsilon\geq 0$, we have
	\begin{align}
		|(1-e^\epsilon)-&\|\bar\omega-e^\epsilon\bar\nu\|_2|
		=
		\|\bar\omega-e^\epsilon\bar\nu\|_2-(1-e^\epsilon).\label{eqn:proof:1}
	\end{align}
	We analyze the other term for the following two cases.
	\begin{itemize}
		\item Case~I: Assume that $\|\bar\omega-e^\epsilon\bar\nu\|_2\geq -(1-e^\epsilon)$. In this case, we have
		\begin{align}
			|(1\!-\!e^\epsilon)\!+\! \|\bar\omega\!-\!e^\epsilon\bar\nu\|_2|
			=
			\|\bar\omega-e^\epsilon\bar\nu\|_2\!+\!(1\!-\!e^\epsilon).\label{eqn:proof:2}
		\end{align}
		Combining~\eqref{eqn:proof:1} and~\eqref{eqn:proof:2} with~\eqref{eqn:proof:0}, we get
		$\|\bar\rho-e^\epsilon\bar\sigma\|_1=\|\bar\omega-e^\epsilon\bar\nu\|_2,$
		which results in
		\begin{align}
			E_{e^\epsilon}(\mathcal{E}(\rho)\|\mathcal{E}(\sigma))
			=&\frac{1}{2}\|\bar\omega-e^\epsilon\bar\nu\|_2
			+\frac{1}{2}(1-e^\epsilon).
			\label{eqn:proof:3}
		\end{align}
		\item Case~II: Assume that $\|\bar\omega-e^\epsilon\bar\nu\|_2< -(1-e^\epsilon)$. In this case, we have
		\begin{align}
			|(1\!-\!e^\epsilon)\!+\! \|\bar\omega\!-\!e^\epsilon\bar\nu\|_2|
			=
			-\|\bar\omega\!-
			\!e^\epsilon\bar\nu\|_2\!-\!(1\!-\!e^\epsilon).\label{eqn:proof:4}
		\end{align}
		Combining~\eqref{eqn:proof:1} and~\eqref{eqn:proof:4} with~\eqref{eqn:proof:0}, we get
		$\|\bar\rho-e^\epsilon\bar\sigma\|_1
		=-(1-e^\epsilon),$
		which results in
		\begin{align}
			E_{e^\epsilon}(\mathcal{E}(\rho)\|\mathcal{E}(\sigma))
			=0.
			\label{eqn:proof:5}
		\end{align}
	\end{itemize}
	Combining Case~I, i.e.,~\eqref{eqn:proof:3}, and Case~II, i.e.,~\eqref{eqn:proof:5}, shows that 
	\begin{align*}
		E_{e^\epsilon}(\mathcal{E}(\rho)\|\mathcal{E}(\sigma))
		=\max\left\{0,\frac{1}{2}\|\bar\omega-e^\epsilon\bar\nu\|_2+\frac{1}{2}(1-e^\epsilon)\right\}.
	\end{align*}
	and, as a result,
	\begin{align*}
		\sup_{\rho,\sigma\in\S(\H)}E_{e^\epsilon}&(\mathcal{E}(\rho)\|\mathcal{E}(\sigma))
		\\&=\max\left\{0,\max_{\bar\omega,\bar\nu}\frac{1}{2}\|\bar\omega-e^\epsilon\bar\nu\|_2+\frac{1}{2}(1-e^\epsilon)\right\}.
	\end{align*}
	Therefore, Lemma~\ref{lem:iff_QLDP} implies that $\mathcal{E}\in{\rm LDP}_{\epsilon}$ if and only if $\|\bar\omega-e^\epsilon\bar\nu\|_2\leq e^\epsilon-1$ for all $\bar\omega,\bar\nu$.
\end{proof}

With these results in hand, we are ready to present sample complexity bounds for scalar parameter estimation under quantum differential privacy. 

\section{Parameter Estimation Sample Complexity}
\label{sec:results}

We first define sample complexity for parameter estimation based on multiple copies of quantum states. 

\begin{definition}[Sample Complexity] For quantum channel $\mathcal{E}$, the minimum number of samples required for obtaining estimation accuracy of $\alpha>0$ is 
	\begin{align*}
		N_{\alpha}(\mathcal{E})\!= \!\!\!\!\inf_{
				\hat{\lambda}:\mathbb{E}\{\hat\lambda \}=\lambda}
			\!\!\inf
			\Big\{N:\mathbb{E}&\{(\lambda\!-\!\hat{\lambda})^2\}\leq \alpha 
			\\[-1em]&
			\mbox{ based on } \mathcal{E}(\rho_\lambda)^{\otimes N}\!\Big\}.
		\end{align*}
\end{definition}
	
The inner infimum captures the smallest number of samples or copies that meet the pre-specified error bound on parameter estimation, i.e., the smallest member of the set $\{N:\mathbb{E}\{(\lambda-\hat{\lambda})^2\}\leq \alpha \mbox{ based on } \mathcal{E}(\rho_\lambda)^{\otimes N}\}$, while the outer infimum captures the best unbiased estimator. Now, we can present the main result of this paper regarding sample complexity of scalar parameter estimation under quantum differential privacy for qubits. 
	
\begin{theorem} \label{thm:parameter_complexity_qubit}
		Let $\dim(\H)=2$. Assume that $\braket{\partial_\lambda \omega_\lambda}{\omega_\lambda}\neq 0$. Then,
		\begin{align*}
			\frac{C_1}{\alpha (e^\epsilon-1)^2} \leq  \inf_{\mathcal{E}\in{\rm LDP}_{\epsilon}}N_{\alpha}(\mathcal{E}) \leq 
			\frac{C_2(e^\epsilon+1)^2}{\alpha (e^\epsilon-1)^2},
		\end{align*}
		where
		\begin{align*}
			C_1&=\frac{1}{\|\partial_\lambda \omega_\lambda\|_2^2}\left(4+\frac{1}{4}\frac{\|\partial_\lambda \omega_\lambda\|_2^2}{|\braket{\partial_\lambda \omega_\lambda}{\omega_\lambda}|^2} \right)^{-1},\\
			C_2&=\frac{1}{\|\partial_\lambda \omega_\lambda\|_2^2}.
		\end{align*}
		Particularly, $N_{\alpha,\epsilon}=\Theta\left(\alpha^{-1}\epsilon^{-2}\right)$ for $\epsilon\ll 1$.
\end{theorem}
	
\begin{proof}
		\textit{Proving the Lower Bound on $N_{\alpha,\epsilon}$:} We prove three important inequalities that enable  bounding the quantum Fisher information.  For the first inequality, let $\omega=\omega_\lambda$ and 
		$$
		\nu=\omega_\lambda-\underbrace{\frac{\braket{\partial_\lambda \omega_\lambda}{\omega_\lambda}}{\|\partial_\lambda \omega_\lambda\|_2^2}}_{:=\beta}\partial_\lambda \omega_\lambda.
		$$
		We have
		$\|\nu\|_2^2=\|\omega_\lambda\|_2^2-|\braket{\partial_\lambda \omega_\lambda}{\omega_\lambda}|^2/\|\partial_\lambda\omega_\lambda\|_2^2\leq \|\omega_\lambda\|_2^2\leq 1.
		$
		Substituting $\omega$ and $\nu$ in Lemma~\ref{lem:iff_QLDP_qubit} results in
		\begin{align*}
			(e^\epsilon-1)^2
			\geq &\|\beta A\partial_\lambda \omega_\lambda+(1-e^\epsilon)(A\omega_\lambda-\beta A\partial_\lambda \omega_\lambda+c)\|_2^2\\
			=&\|(1-e^\epsilon)(A\omega_\lambda+c)
			+\beta e^\epsilon A\partial_\lambda \omega_\lambda\|_2^2\\
			=&(1-e^\epsilon)^2\|A\omega_\lambda+c\|_2^2
			+\beta^2e^{2\epsilon}\|A\partial_\lambda \omega_\lambda\|_2^2\\
			&-2\beta(e^\epsilon-1)e^{\epsilon}\braket{A\omega_\lambda+c}{A\partial_\lambda \omega_\lambda}\\
			\geq& (1-e^\epsilon)^2\|A\omega_\lambda+c\|_2^2\\
			&-2\beta(e^\epsilon-1)e^{\epsilon}\braket{A\omega_\lambda+c}{A\partial_\lambda \omega_\lambda}.
		\end{align*}
		Noting that $e^\epsilon\geq 1$, we get
		\begin{align}
			-\beta\braket{A\omega_\lambda+c}{A\partial_\lambda \omega_\lambda}
			&\leq \frac{1}{e^\epsilon} \frac{e^\epsilon-1}{2}(1-\|A\omega_\lambda+c\|_2^2)\nonumber\\
			&\leq \frac{e^\epsilon-1}{2}(1-\|A\omega_\lambda+c\|_2^2).
			\label{eqn:inequality_proof_1}
		\end{align}
		For the second inequality, let $\nu=\omega_\lambda$ and 
		$$
		\omega=\omega_\lambda-\underbrace{\frac{\braket{\partial_\lambda \omega_\lambda}{\omega_\lambda}}{\|\partial_\lambda \omega_\lambda\|_2^2}}_{:=\beta}\partial_\lambda \omega_\lambda.
		$$
		We have
		$\|\omega\|_2^2=\|\omega_\lambda\|_2^2-|\braket{\partial_\lambda \omega_\lambda}{\omega_\lambda}|^2/\|\partial_\lambda\omega_\lambda\|_2^2\leq \|\omega_\lambda\|_2^2\leq 1.
		$
		Substituting $\omega$ and $\nu$ in Lemma~\ref{lem:iff_QLDP_qubit} results in
		\begin{align*}
			(e^\epsilon-1)^2
			\geq &\|-\beta A\partial_\lambda \omega_\lambda+(1-e^\epsilon)(A\omega_\lambda+c)\|_2^2\\
			=&(1-e^\epsilon)^2\|A\omega_\lambda+c\|_2^2
			+\beta^2\|A\partial_\lambda \omega_\lambda\|_2^2\\
			&+2\beta(e^\epsilon-1)\braket{A\omega_\lambda+c}{A\partial_\lambda \omega_\lambda}\\
			\geq& (1-e^\epsilon)^2\|A\omega_\lambda+c\|_2^2\\
			&+2\beta(e^\epsilon-1)\braket{A\omega_\lambda+c}{A\partial_\lambda \omega_\lambda}.
		\end{align*}
		Thus
		\begin{align}
			\beta\braket{A\omega_\lambda+c}{A\partial_\lambda \omega_\lambda}
			&\leq \frac{e^\epsilon-1}{2}(1-\|A\omega_\lambda+c\|_2^2).
			\label{eqn:inequality_proof_2}
		\end{align}
		Combining~\eqref{eqn:inequality_proof_1} and~\eqref{eqn:inequality_proof_2} while recalling definition of $\beta$, we get
		\begin{align*}
			\Bigg|\frac{\braket{\partial_\lambda \omega_\lambda}{\omega_\lambda}}{\|\partial_\lambda \omega_\lambda\|_2^2}&\braket{A\omega_\lambda+c}{A\partial_\lambda \omega_\lambda}\Bigg|
			\!\leq\! \frac{e^\epsilon\!-\!1}{2}(1\!-\!\|A\omega_\lambda+c\|_2^2),
		\end{align*}
		and hence
		\begin{align}
			\frac{|\braket{A\omega_\lambda+c}{A\partial_\lambda \omega_\lambda}|}{(1-\|A\omega_\lambda+c\|_2^2)}
			\leq \frac{e^\epsilon-1}{2}\frac{\|\partial_\lambda \omega_\lambda\|_2^2}{|\braket{\partial_\lambda \omega_\lambda}{\omega_\lambda}|}.
			\label{eqn:inequality_proof_1_and_2}
		\end{align}
		For the third inequality, let $\nu=0$. We get $e^\epsilon-1\geq \|A\omega+(1-e^\epsilon)c\|_2\geq \|A\omega\|_2-(e^\epsilon-1)\|c\|_2$ and thus, it must be that $\|A\omega\|_2\leq (e^\epsilon-1)(1+\|c\|)\leq 2(e^\epsilon-1)$, which gives
		\begin{align}\label{eqn:inequality_proof_3}
			\|A\|_2=\sup_{\|\omega\|_2=1}\|A\omega\|_2
			\leq 2(e^\epsilon-1).
		\end{align}
		Now, we are ready to bound quantum Fisher information:
		\begin{align*}
			\mathcal{F}(\mathcal{E}(\rho_\lambda))
			\leq & \|A\partial_\lambda \omega_\lambda\|_2^2+\displaystyle \frac{|\braket{A\omega_\lambda+c}{A\partial_\lambda \omega_\lambda}|^2}{1-\|A\omega_\lambda+c\|_2^2}\nonumber\\
			\leq & 4(e^\epsilon-1)^2\|\partial_\lambda \omega_\lambda\|_2^2
			\nonumber\\
			&+\frac{(e^\epsilon-1)^2}{4}\frac{\|\partial_\lambda \omega_\lambda\|_2^4}{|\braket{\partial_\lambda \omega_\lambda}{\omega_\lambda}|^2}(1-\|A\omega_\lambda+c\|_2^2)\\
			\leq & 4(e^\epsilon-1)^2\|\partial_\lambda \omega_\lambda\|_2^2
			\left(1+\frac{1}{16}\frac{\|\partial_\lambda \omega_\lambda\|_2^2}{|\braket{\partial_\lambda \omega_\lambda}{\omega_\lambda}|^2} \right).
		\end{align*}
		Finally, from the quantum Cram\'{e}r-Rao bound, we get
		\begin{align*}
			\alpha\geq &\mathbb{E}\{(\lambda-\hat{\lambda})^2\}\\
			\geq& \frac{1}{N }\frac{1}{4(e^\epsilon-1)^2\|\partial_\lambda \omega_\lambda\|_2^2}\left(1+\frac{1}{16}\frac{\|\partial_\lambda \omega_\lambda\|_2^2}{|\braket{\partial_\lambda \omega_\lambda}{\omega_\lambda}|^2} \right)^{-1}.
		\end{align*}
		
		\textit{Proving the Upper Bound on $N_{\alpha,\epsilon}$:} Select $\mathcal{E}(\rho)=\frac{p}{2}I+(1-p)\rho.$
		This is the so-called global depolarizing channel. From Lemma~IV.2 in~\cite{hirche2023quantum}, we know that $\mathcal{E}\in{\rm LDP}_\epsilon$ if $p=2/(1+e^\epsilon)$. We have $\mathcal{F}(\mathcal{E}(\rho_\lambda))
		\geq   (1-p)^2\|\partial_\lambda \omega_\lambda\|_2^2
		= ((e^\epsilon-1)^2/(e^\epsilon+1)^2)\|\partial_\lambda \omega_\lambda\|_2^2.$
		From\cite{braunstein1994statistical, braunstein1996generalized}, we know that, in the case of scalar parameters, there exists an unbiased estimator for which the quantum Cram\'{e}r-Rao bound is saturated. That is, $\alpha=\mathbb{E}\{(\lambda-\hat{\lambda})^2\}=1/(N \mathcal{F}(\mathcal{E}(\rho_\lambda))).$ Therefore, for this specific unbiased estimator, we get $
		N\leq ((e^\epsilon+1)^2/(e^\epsilon-1)^2) /(\alpha\|\partial_\lambda \omega_\lambda\|_2^2).$
		This concludes the proof.
\end{proof}
	
Note that Theorem~\ref{thm:parameter_complexity_qubit} bounds the best attainable sample complexity, i.e., the number of samples required for the best\footnote{\label{foot:1}The problem of characterizing the worst-case sample complexity of parameter estimation under quantum differential privacy is ill-defined. To observe this, note that a constant quantum channel $\mathcal{E}:\rho_\lambda\mapsto(I/d)$ is differentially-private for all $\epsilon>0$ but, irrespective of how many copies of $\mathcal{E}(\rho_\lambda)$ is available, the underlying parameter $\lambda$ never can  be recovered.} differentially-private channel. This result shows that, as the privacy guarantees strengthen $\epsilon\rightarrow 0$, the sample complexity grows quadratically with $\epsilon$. That is, if privacy budget is halved, we need four times as many samples or copies of the quantum state to be able to attain the same level of accuracy. This observation provides a concrete interpretation for the privacy budget in quantum differential privacy.
	
\begin{corollary} \label{cor:1}
		Under the assumptions of Theorem~\ref{thm:parameter_complexity_qubit},
		\begin{align*}
			\frac{C_1}{9\alpha\epsilon^2} \leq  \inf_{\mathcal{E}\in{\rm LDP}_{\epsilon}}N_{\alpha}(\mathcal{E}) \leq 
			\frac{C_2(e+1)^2}{\alpha\epsilon^2},
		\end{align*}
		for all $\epsilon\in(0,1)$.
\end{corollary}
	
\begin{proof}
		Note that $e^\epsilon\geq 1+\epsilon, \forall \epsilon>0$~\cite[p.\,81, Eq.~(3)]{bullen2015dictionary} and $e^\epsilon\leq (1+\epsilon)^2, \forall \epsilon\in(0,1)$~\cite[p.\,81, Eq.~(2)]{bullen2015dictionary}. Therefore, $\epsilon^2\leq (e^\epsilon-1)^2\leq \epsilon^2(2+\epsilon)^2\leq 9\epsilon^2$. Also, $e^\epsilon\leq e$, which gives $(e^\epsilon+1)^2\leq (e+1)^2$. Combining these inequalities with the result of Theorem~\ref{thm:parameter_complexity_qubit} completes the proof. 
\end{proof}
		
Corollary~\ref{cor:1} shows that the best-case sample complexity $\inf_{\mathcal{E}\in{\rm LDP}_{\epsilon}}N_{\alpha}(\mathcal{E})$ scales with $\alpha^{-1}\epsilon^{-2}$ over the entire range of $0<\epsilon<1$ (and not just $\epsilon\ll 1$). Therefore, the developed sample complexity bounds are tight up to a constant for small and medium privacy budget regimes, captured by $0<\epsilon<1$. 

\begin{remark}
	Note that privacy analysis is adversarial in nature, and thus the estimator and the data curator are at odds, i.e., the curator does not want to simplify the task of the estimator. This motivates fixing the description of the quantum state $\rho_\lambda$ by the problem setting (e.g., encoding in quantum communication or machine learning) and the risk appetite of the curator. However, in general, we can consider the impact of the quantum state $\rho_\lambda$ on the sample complexity. For instance, one might
	design $\rho_\lambda$ so as to maximize the quantum Fisher information in the hope of achieving a better sample complexity even after corruption via quantum channels. To illustrate the impact of such a design, consider  pure state $\rho_\lambda$, i.e., $\|\omega_\lambda\|_2=1$. In this case, we can easily see that $C_1\propto \mathcal{F}(\rho_\lambda)^{-1}$ and $C_2=\mathcal{F}(\rho_\lambda)^{-1}$. Therefore, by maximizing the quantum Fisher information, we merely make constants $C_1$ and $C_2$ smaller but the structural results regarding the dependence of sample complexity on $\epsilon$ and $\alpha$ remain the same. This illustrates that quantum  differential privacy restricts the worst-case information flow. 
\end{remark}
			
\begin{remark}
	In the large privacy budget regime, i.e., $\epsilon\gg 1$, The sample complexity bound in Theorem~\ref{thm:parameter_complexity_qubit} simplifies to
	$(C_1/\alpha) e^{-2\epsilon} \lessapprox \inf_{\mathcal{E}\in{\rm LDP}_{\epsilon}} N_{\alpha}(\mathcal{E}) \lessapprox
	(C_2/\alpha).$ Therefore, the upper and lower bounds no longer admit the same asymptotic behavior and are thus no longer tight. In fact, the lower bound  loosens arbitrarily, i.e., the lower bound states that, in the non-private regime $\epsilon\rightarrow \infty$, the number of samples required to achieve a certain level of accuracy tends to zero, which deviates from the established bounds in the absence of quantum differential privacy~\cite{meyer2023quantum}. 
\end{remark}
		
Theorem~\ref{thm:parameter_complexity_qubit}, and by extension Corollary~\ref{cor:1}, assume that $\braket{\partial_\lambda \omega_\lambda}{\omega_\lambda}\neq 0$. This cannot be guaranteed when $\omega_\lambda$ is a pure state for all $\lambda$. To observe this, note that $\braket{\partial_\lambda \omega_\lambda}{\omega_\lambda}=\partial \|\omega_\lambda\|_2^2/\partial \lambda =0$ if $\|\omega_\lambda\|_2=1$ for all $\lambda$. In this case, we can derive matching lower bounds for the following restricted set of quantum channels $\overline{{\rm LDP}}_{\epsilon}:=\{(A,c)\in{\rm LDP}_{\epsilon}|c=0\}.$ These channels capture many noise models that have shown to satisfy quantum differential privacy~\cite{hirche2023quantum}. 
		
\begin{theorem} \label{tho:2}
	Let $\dim(\H)=2$. Then, for all $\epsilon\in(0,1/2)$, 
	\begin{align*}
		\frac{\overline{C}_1}{9\alpha\epsilon^2}\leq 
		\frac{\overline{C}_1}{\alpha(e^\epsilon-1)^2}\leq 
		\inf_{\mathcal{E}\in\overline{{\rm LDP}}_{\epsilon}}N_{\alpha}(\mathcal{E}) \leq 
		\frac{C_2(\sqrt{e}+1)^2}{\alpha \epsilon^2},
	\end{align*}
	where
	\begin{align*}
		\overline{C}_1&=\frac{1}{\|\partial_\lambda \omega_\lambda\|_2}\frac{1}{\|\partial_\lambda \omega_\lambda\|_2+(\sqrt{e}(2-\sqrt{e}))^{-1}},
	\end{align*}
	and $C_2$ is defined in Theorem~\ref{thm:parameter_complexity_qubit}. 
\end{theorem}
		
\begin{proof} 
	Setting $\nu=0$ and $\omega=\omega_\lambda$ in Lemma~\ref{lem:iff_QLDP_qubit} results in $\|A\omega\|_2\leq e^\epsilon-1$ and so $\|A\|_2=\sup_{\|\omega\|_2=1} \|A\omega\|_2 \leq e^\epsilon-1.$ Therefore, we can deduce that $\|A\partial_\lambda \omega_\lambda\|_2\leq (e^\epsilon-1)\|\partial_\lambda \omega_\lambda\|_2^2$. Now, we are ready to bound quantum Fisher information as
	\begin{align*}
		\mathcal{F}(\mathcal{E}(\rho_\lambda))
		\leq & \|A\partial_\lambda \omega_\lambda\|_2^2+\displaystyle \frac{|\braket{A\omega_\lambda}{A\partial_\lambda \omega_\lambda}|^2}{1-\|A\omega_\lambda\|_2^2}\nonumber\\
		\leq & \|A\partial_\lambda \omega_\lambda\|_2^2+\displaystyle \frac{\|A\omega_\lambda\|_2\|A\partial_\lambda \omega_\lambda\|_2}{1-\|A\omega_\lambda\|_2^2}\nonumber\\
		\leq & (e^\epsilon-1)^2\|\partial_\lambda \omega_\lambda\|_2^2+\frac{(e^\epsilon-1)^2\|\partial_\lambda \omega_\lambda\|_2}{1-(e^\epsilon-1)^2}\\
		\leq & (e^\epsilon-1)^2\|\partial_\lambda \omega_\lambda\|_2 \left(\!\|\partial_\lambda \omega_\lambda\|_2\!+\!\frac{1}{\sqrt{e}(2-\sqrt{e})} 
		\!\right),
	\end{align*}
	where the last inequality follows from that $1-(e^\epsilon-1)^2=e^\epsilon(2-e^\epsilon)\geq \sqrt{e}(2-\sqrt{e})$ for  $0<\epsilon<1/2$. The rest follows from the quantum Cram\'{e}r-Rao bound and the observation that $(e^\epsilon-1)^2\leq 9\epsilon^2$; see Corollary~\ref{cor:1}. The proof for the upper bound is the same as Corollary~\ref{cor:1} except $\epsilon\leq 1/2$. 
\end{proof}

\begin{proposition}
	Let $\dim(\H)=d$. Then $N_{\alpha,\epsilon}=\mathcal{O}(d\alpha^{-1}\epsilon^{-2})$ for $\epsilon\ll 1$.
\end{proposition}
	
\begin{proof}
	For qudits, the global depolarizing channel gives $\mathcal{E}(\rho_\lambda)=\frac{p}{d}I+(1-p)\rho_\lambda=\frac{1}{d}I+\frac{1}{2}(1-p)\omega_\lambda.\hat\eta.$ 
	From Lemma~IV.2 in~\cite{hirche2023quantum}, we know that $\mathcal{E}\in{\rm LDP}_\epsilon$ if $p=d/(d-1+e^\epsilon)$. 
	Note that if $\epsilon\ll 1/d$, $(1-p)\omega_\lambda=\mathcal{O}(\epsilon/d)$. Therefore, $\mathcal{M}(\omega_\lambda)=(2/d)I+\mathcal{O}(\epsilon/d)$. We have $\mathcal{F}(\mathcal{E}(\rho_\lambda))\geq  (1-p)^2(d/2)\|\partial_\lambda \omega_\lambda\|_2^2=\Omega(\epsilon^2/d).$ Therefore, there exists an estimator for which we can attain the desired estimation error if $N=\mathcal{O}(d\alpha^{-1}\epsilon^{-2})$. 
\end{proof}
	
The dependence of the upper bound to $\epsilon$ does not change when moving from qubits to qudits. The linear dependence to the dimension is not known to be tight. An avenue for future research is to establish matching lower bounds on sample complexity to investigate optimal dependence on dimension. 
	
\begin{remark}
	A similar bound for the classical case is shown to hold~\cite{barnes2020fisher}. For scalar parameters, Proposition~2 and Remark~1 in~\cite{barnes2020fisher} show that the sample complexity in the classical case scales as $\mathcal{O}(\alpha^{-1}\epsilon^{-2})$, which takes the form of Theorem~\ref{thm:parameter_complexity_qubit}, Corollary~\ref{cor:1}, Theorem~\ref{tho:2}, and Proposition~\ref{thm:parameter_complexity_qubit} in the small privacy budget regime. The dependence of the bound on the dimension of the quantum system in Proposition~\ref{thm:parameter_complexity_qubit} distinguishes the bounds in the quantum case from the classical case. 
\end{remark}

\section{Conclusions}\label{sec:conc}
This paper presented upper and lower bounds for sample complexity of parameter estimation under quantum local differential privacy. We particularly consider the best-case scenario in which the sample complexity is minimized over all differentially-private channels. The focus is on scalar parameter regime where it is known that the quantum Cram\'{e}r-Rao can be saturated with finitely-many samples. The bounds match the classical parameter estimation under local differential privacy. Future work can focus on multi-parameter regime and biased estimators. Extension to multi-parameter setting can  start with the sub-regime that guarantees saturation of the quantum Cram\'{e}r-Rao bound~\cite{10480691,ragy2016compatibility}.

\section*{Acknowledgment}
The author would like to thank Theshani Nuradha for pointing out a mathematical typo in the proof of Theorem~\ref{thm:parameter_complexity_qubit}.

\bibliography{ref}
\bibliographystyle{ieeetr}

\appendix

\section{Useful Lemma}

\setcounter{lemma}{0}
\renewcommand{\thelemma}{\Alph{lemma}}

\begin{lemma}\label{lem:useful}
	For any $m\in\R$ and $n\in\R^3$,
	$\|mI+n.\hat{\sigma}\|_1=|m-\|n\|_2|+|m+\|n\|_2|$.
\end{lemma}

\begin{proof}
	Note that $\det(mI+n.\hat{\sigma}-sI)=(s-m)^2-\|n\|_2^2.$ Therefore, the eigenvalues of $mI+n.\hat{\sigma}$ are $s_{\pm}=m\pm \|n\|_2$ and so $\|mI+n.\hat{\sigma}\|_1=|s_+|+|s_-|$.
\end{proof}
	
\end{document}